\documentclass[twoside]{article}

\usepackage[USenglish]{babel}
\usepackage{amsmath,amssymb,amsthm,url,braket,hyperref,cite,subfig,algorithm,algorithmic}

\DeclareMathOperator{\Sym}{Sym}
\DeclareMathOperator{\Geom}{Geom}
\DeclareMathOperator{\Perm}{Perm}

\DeclareMathOperator{\Tr}{Tr}

\newtheorem{lmm}{Lemma}

\begin{document}

\setlength{\textheight}{8.0truein}


\normalsize


\thispagestyle{empty}

\setcounter{page}{1}


\vspace*{0.88truein}



\centerline{\bf COMPUTING THE QUANTUM GUESSWORK}

\vspace*{0.035truein}

\centerline{\bf A QUADRATIC ASSIGNMENT PROBLEM}

\vspace*{0.37truein}

\centerline{\footnotesize Michele Dall'Arno}

\vspace*{0.015truein}

\centerline{\footnotesize\it Department  of Computer Science
  and  Engineering,  Toyohashi   University  of  Technology}

\baselineskip=10pt

\centerline{\footnotesize\it  1-1 Hibarigaoka,  Tempaku-cho,
  Toyohashi, Aichi, 441-8580, Japan}

\vspace*{0.015truein}

\centerline{\footnotesize\it     Yukawa    Institute     for
  Theoretical Physics, Kyoto University}

\baselineskip=10pt

\centerline{\footnotesize\it   Sakyo-ku,  Kyoto,   606-8502,
  Japan}

\vspace*{10pt}

\centerline{\footnotesize Francesco Buscemi}

\vspace*{0.015truein}

\centerline{\footnotesize\it   Department  of   Mathematical
  Informatics,  Graduate   School  of   Informatics,  Nagoya
  University}

\baselineskip=10pt

\centerline{\footnotesize\it   Chikusa-Ku,  Nagoya,   Aichi,
  464-8601, Japan}

\vspace*{10pt}

\centerline{\footnotesize Takeshi Koshiba}

\vspace*{0.015truein}

\centerline{\footnotesize\it   Faculty   of  Education   and
  Integrated Arts and Sciences, Waseda University}

\baselineskip=10pt

\centerline{\footnotesize\it  Shinjuku-ku, Tokyo,  169-8050,
  Japan}

\vspace*{0.225truein}


\vspace*{0.21truein}

\begin{abstract}
  The  quantum guesswork  quantifies the  minimum number  of
  queries needed to guess the state of a quantum ensemble if
  one  is  allowed  to  query  only one  state  at  a  time.
  Previous approaches  to the  computation of  the guesswork
  were   based   on   standard   semi-definite   programming
  techniques and therefore lead to approximated results.  In
  contrast, we show that  computing the quantum guesswork of
  qubit  ensembles  with  uniform  probability  distribution
  corresponds to solving a  quadratic assignment problem and
  we provide an algorithm that,  upon the input of any qubit
  ensemble over  a discrete ring, after  finitely many steps
  outputs the exact closed-form expression of its guesswork.
  While in general the complexity of our guesswork-computing
  algorithm is factorial  in the number of  states, our main
  result consists  of showing a  more-than-quadratic speedup
  for symmetric  ensembles, a scenario corresponding  to the
  three-dimensional  analog of  the maximization  version of
  the turbine-balancing  problem.  To find  such symmetries,
  we provide an algorithm that,  upon the input of any point
  set  over  a  discrete  ring, after  finitely  many  steps
  outputs  its  exact  symmetries.  The  complexity  of  our
  symmetries-finding algorithm  is polynomial in  the number
  of  points.   As examples,  we  compute  the guesswork  of
  regular and quasi-regular sets of qubit states.
\end{abstract}

\vspace*{10pt}


\vspace*{3pt}


\vspace*{3pt}

{\centering{
  \begin{minipage}{4.5in}
    \footnotesize
    \baselineskip=10pt
        {\footnotesize\it Report number}\/: YITP-21-150
\end{minipage}}
\par}

\vspace*{1pt}

\section{Introduction}

We consider the  following communication scenario, described
in  terms  of  standard  concepts  and  results  in  quantum
information theory~\cite{Wil17}.  Let an ensemble of quantum
states be given.  At each  round, a referee prepares a state
from the ensemble.  The task is  to guess which state it is,
being  allowed  to query  one  state  at  a time  until  the
referee's answer is on the affirmative, at which point a new
round  begins.   The cost  function  is  represented by  the
average  number of  queries  needed to  correctly guess  the
state of the  ensemble, and is therefore referred  to as the
\textit{quantum guesswork}~\cite{Mas94, Ari96, AM98a, AM98b,
  Pli98,  MS04,  Sun07,  HS10, CD13,  SV18,  Sas18,  CCWF15,
  HKDW20, DBK22}.  Notice that,  if multiple states could be
queried  at a  time, the  corresponding cost  function would
instead be the entropy~\cite{HKDW20} of the ensemble.

The  most  general  strategy   consists  of  a  sequence  of
nondemolishing  quantum measurements  (quantum instruments),
each producing as a classical outcome the next query for the
referee. However, as further detailed in Ref.~\cite{HKDW20},
by considering  the composition of such  instruments, such a
strategy reduces to performing  a single quantum measurement
(on the single  copy of the state provided  by the referee),
whose  classical  outcomes  are  represented  by  tuples  of
ordered queries for  the referee.  In other  words, one will
first query the  referee for the state  corresponding to the
first entry  in the output  tuple. If  the answer is  on the
negative, one will proceed querying  for the second entry in
the  output tuple,  and so  on, with  the goal  of correctly
guessing with the minimum number of queries.

The usual approach~\cite{CCWF15} to compute the guesswork is
based on a factorial-size semidefinite program, outputting a
numerically approximated  result within any  given tolerance
in  polynomial time  in  the  (factorially growing)  problem
size. Here, instead, we are interested in an exact algorithm
to  compute   the  guesswork  in  finite   time,  where  our
computational model is a  machine capable of storing integer
numbers and  of performing  additions and  multiplication in
finite time.  Even the existence of such an algorithm is not
guaranteed a priori, given that the guesswork problem is, by
definition,  a  continuous  optimization problem.   For  the
qubit case, however, the equivalence of the guesswork with a
finite    optimization    problem    has    recently    been
shown~\cite{DBK22}.    Our   analysis    begins   with   the
observation that  such a  finite optimization problem  is an
instance  of the  (generally NP-hard~\cite{SG76})  quadratic
assignment   problem~\cite{KB57}    (see   also~\cite{Cel98,
  BCPP98} for reviews).

Our main  result consists in  showing how the  symmetries of
the ensemble, for whose characterization we provide an exact
polynomial-time  algorithm, can  be exploited  to achieve  a
more-than-quadratic  speedup  in   the  computation  of  the
guesswork.     This    scenario     corresponds    to    the
three-dimensional analog of the  maximization version of the
turbine-balancing   problem~\cite{LM88}  for   a  particular
vector  of   coefficients,  in  which  blades   are  ideally
symmetrically distributed  on a sphere instead  of a circle.
To illustrate our results, we provide implementations of our
symmetries-finding and guesswork-computing algorithms in the
C programming language,  and we use them  to exactly compute
the guesswork  of regular and quasi-regular  ensembles of up
to  twenty-four   states,  geometrically   corresponding  to
Platonic and Archimedean solids in the Bloch sphere.

\section{Formalization}
\label{sect:guesswork}

Qubit  states are  in one-to-one  correspondence with  Pauli
vectors, that is, three-dimensional  vectors within the unit
sphere.  Hence, an ensemble of $N$ qubit states with uniform
probability  distribution  and  without repetitions  can  be
represented  by  a finite  set  $\mathcal{V}$  of $N$  Pauli
vectors, that is
\begin{align*}
  \mathcal{V} \subseteq \left\{ v  \in \mathbb{R}^3 \; \Big|
  \; \left| v \right|_2 \le 1 \right\}.
\end{align*}
Following     Ref.~\cite{Knu11},      we     denote     with
$\mathcal{V}^{\underline{N}}$  the  set   of  $N$-tuples  on
$\mathcal{V}$ without  repetitions, that is  $\mathbf{v} \in
\mathcal{V}^{\underline{N}}$ if and only if
\begin{align*}
  \mathbf{v} = \left( v_i \right)_{i \in \left\{ 1, \dots, N
    \right\}}  \textrm{ s.t.   } \bigcup_{i  \in \left\{  1,
    \dots, N \right\}} \left\{ v_i \right\} = \mathcal{V}.
\end{align*}
Any qubit  effect can be  represented as an  affine function
from  the set  of  qubit  states to  $[0,  1]$.  Hence,  any
$\mathcal{V}^{\underline{N}}$  valued qubit  measurement can
be represented as a family $( \pi_{\mathbf{v}} )_{\mathbf{v}
  \in \mathcal{V}^{\underline{N}}}$ of such affine functions
indexed by $\mathcal{V}^{\underline{N}}$ such
that   $\sum_{\mathbf{v}  \in   \mathcal{V}^{\underline{N}}}
\pi_{\mathbf{v}} = u$, where $u$ is the unit effect, that is
\begin{align*}
  \mathcal{P}  = \left(  \pi_{\mathbf{v}} :  \mathcal{S} \to
  \left[     0,1     \right]     \right)_{\mathbf{v}     \in
    \mathcal{V}^{\underline{N}}}     \textrm{    s.t.      }
  \sum_{\mathbf{v}      \in     \mathcal{V}^{\underline{N}}}
  \pi_{\mathbf{v}} = u.
\end{align*}

For   any  such   qubit  ensemble   $\mathcal{V}$  and   any
$\mathcal{V}^{\underline{N}}$-valued    qubit    measurement
$\mathcal{P}$,   the  \textit{guesswork}   $G  (\mathcal{V},
\mathcal{P})$ is the number of  queries needed on average to
correctly guess  the state  of $\mathcal{V}$  when measuring
$\mathcal{P}$,  that  is  (see  Ref.~\cite{DBK22}  for  more
details)
\begin{align*}
  G  \left( \mathcal{V},  \mathcal{P}  \right) :=  \frac1{N}
  \sum_{\substack{\mathbf{v}                             \in
      \mathcal{V}^{\underline{N}}\\ i \in  \left\{ 1, \dots,
      N \right\}}} \pi_{\mathbf{v}} \left( v_i \right) \; i,
\end{align*}
where $v_i$ denotes the $i$-th  component of vector $v$. The
\textit{minimum  guesswork}   $G_{\min}  \left(  \mathcal{V}
\right)$ is the  minimum of the guesswork  $G ( \mathcal{V},
\mathcal{P} )$ over any $\mathcal{V}^{\underline{N}}$-valued
measurement $\mathcal{P}$, that is
\begin{align}
  \label{eq:gw}
  G_{\min}      \left(      \mathcal{V}      \right)      :=
  \min_{\mathcal{P}} G  \left( \mathcal{V}, \mathcal{P}
  \right).
\end{align}

\section{Main result}

Our first  observation is  that the optimization  problem in
Eq.~\eqref{eq:gw} corresponds to a  specific instance of the
quadratic assignment  problem. Since the former  consists of
an optimization over a continuous set, the latter represents
its closed-form solution.

\begin{lmm}
  \label{thm:guesswork}
  For any  ensemble $\mathcal{V}$ of $N$  Pauli vectors with
  uniform  probability distribution,  the minimum  guesswork
  $G_{\min}  (  \mathcal{V} )$  is  given  by the  following
  quadratic assignment problem
  \begin{align*}
    G_{\min} \left( \mathcal{V} \right) = \frac12 \left( N +
    1 - \frac1N  \sqrt{ \max_{X \in \Perm  \left( N \right)}
      \Tr \left[ W X D X^T \right]} \right),
  \end{align*}
  where  $\Perm(N)$   denotes  the  set  of   $N  \times  N$
  permutation matrices, $D$ is the $N \times N$ matrix given
  by $D = d d^T$, in  turn $d$ is the $N$-dimensional column
  vector whose $i$-th entry is given by $2i - N + 1$, $W$ is
  the $N \times N$ Gram matrix given by $W := V^T V$, and in
  turn  $V$ is  any $3  \times N$  matrix whose  columns are
  given   by   the   elements   of   $\mathcal{V}$   without
  repetitions.
\end{lmm}

\begin{proof}
  From  Theorem 1  and Corollary  1 of  Ref.~\cite{DBK22} it
  immediately follows that
  \begin{align}
    \label{eq:optimization}
    G_{\min} \left( \mathcal{V} \right) = \frac12 \left( N +
    1    -     \frac1N    \sqrt{     \max_{\mathbf{v}    \in
        \mathcal{V}^{\underline{N}}}   \left|   \left\langle
      \mathbf{v} \right\rangle \right|^2_2} \right),
  \end{align}
  where  for  any  $N$-tuple  $\mathbf{v}$  of  elements  of
  $\mathcal{V}$  without repetitions  the function  $\langle
  \mathbf{v} \rangle$ is given by
  \begin{align}
    \label{eq:avg}
    \left\langle \mathbf{v} \right\rangle  := \sum_{i = 1}^N
    \left( 2 i - N + 1 \right) v_i.
 \end{align}
  The  result then  follows by  observing that,  by explicit
  computation, one has
  \begin{align*}
    \max_{\mathbf{v} \in \mathcal{V}^{\underline{N}}} \left|
    \left\langle  \mathbf{v}   \right\rangle  \right|^2_2  =
    \max_{X \in \Perm \left( N \right)} \Tr \left[ W X D X^T
      \right].
  \end{align*}
\end{proof}

To begin with, let us analyze the complexity of solving this
quadratic assignment problem with a naive exhaustive search.
The  complexity of  the exhaustive  search is  given by  the
product  of   the  complexity   of  generating   all  tuples
$\mathbf{v}$   in   $\mathcal{V}^{\underline{N}}$  and   the
complexity  of computing  the  function $\langle  \mathbf{v}
\rangle$ for  each tuple.  According  to Eq.~\eqref{eq:avg},
the complexity of computing  $\langle \mathbf{v} \rangle$ is
$O(N)$.  This can be improved if the tuples are generated so
that each  new tuple  is obtained from  the previous  one by
swapping only  two elements.  This is  possible, for example
by     means     of      Heap's     or     Johnson-Trotter's
algorithm~\cite{Sed77,  Knu11}.   In   this  case,  $\langle
\mathbf{v} \rangle$ does not  need be entirely recomputed at
each  step following  Eq.~\eqref{eq:avg}: its  value can  be
stored,  and after  each  new tuple  is generated,  $\langle
\mathbf{v}  \rangle$  can be  updated  by  only taking  into
account the contributions of the  two states affected by the
swapping.  From  Eq.~\eqref{eq:avg}, if  tuples $\mathbf{v}$
and $\mathbf{v}'$ differ by the swap of elements $i$ and $j$
with $i < j$, then we have
\begin{align}
  \label{eq:update}
  \left\langle  \mathbf{v}'   \right\rangle  =  \left\langle
  \mathbf{v} \right\rangle  + 2  \left(j - i  \right) \left(
  v_i - v_j \right),
\end{align}
that   represents  a   constant-time   update  of   $\langle
\mathbf{v} \rangle$.  Hence, the complexity of the algorithm
is  $O(N!)$.   Since  this  value also  corresponds  to  the
cardinality  of set  $\mathcal{V}^{\underline{N}}$, it  also
represents the optimal complexity  for an exhaustive search.
By denoting  with $\operatorname{nextJT}$ the  function that
returns the pair of indexes $(i, j)$ that need to be swapped
to  generate   the  next   tuple  according  to   Heap's  or
Johnson-Trotter     algorithm~\cite{Sed77,    Knu11}     and
$\operatorname{false}$ if the end  of the algorithm has been
reached, we arrive at Algorithm~\ref{algo:gw}.

\begin{algorithm}
\caption{Guesswork (naive exhaustive search)}
\label{algo:gw}
\begin{algorithmic}
  \REQUIRE  finite set $\mathcal{V}$  of $N$  Pauli  vectors
  \ENSURE $G_{\min} ( \mathcal{V} ) = (N + 1 - \sqrt{g}/N)/2$
  \STATE $\mathbf{v} \in \mathcal{V}^{\underline{N}}$
  \STATE $\overline{v} \gets \langle \mathbf{v} \rangle$
  \STATE $g \gets 0$

  \WHILE {$(i, j) \gets \operatorname{nextJT}(\mathbf{v})$}
  \STATE $v_i \leftrightarrow v_j$
  \STATE $\overline{v} \gets \overline{v}  + 2 (j-i) (v_i - v_j)$
  \STATE $g \gets \max \left( g, | \overline{v} |_2^2 \right)$
  \ENDWHILE
  
  \RETURN $g$
\end{algorithmic}
\end{algorithm}

Our main result consists in showing that, in the presence of
symmetries   and  specifically   if  set   $\mathcal{V}$  is
centrally symmetric or vertex  transitive, the complexity of
Algorithm~\ref{algo:gw}    can     be    reduced     by    a
more-than-quadratic factor. This scenario corresponds to the
three-dimensional analog of the  maximization version of the
turbine-balancing   problem~\cite{LM88}  for   a  particular
vector  of   coefficients,  in  which  blades   are  ideally
symmetrically distributed on a sphere instead of a circle.

A set $\mathcal{V}$ of  Pauli vectors is centrally symmetric
if and  only if  for any  $\mathbf{v} \in  \mathcal{V}$ also
$-\mathbf{v}  \in  \mathcal{V}$.   A  set  $\mathcal{V}$  is
vertex  transitive  if   and  only  if  for   any  pair  $\{
\mathbf{v}_0, \mathbf{v}_1 \}  \subseteq \mathcal{V}$, there
exists  orthogonal matrix  $O$ such  that $\mathbf{v}_1  = O
\mathbf{v}_0$  and   $\mathcal{V}  =  O   \mathcal{V}$.   We
postpone to Section~\ref{sect:symmetries}  the discussion of
a  polynomial-time  algorithm   that  exactly  computes  the
symmetries of $\mathcal{V}$.

For any $N$-tuple  $\mathbf{v} = (v_1, \dots,  v_N)$, let us
define  the  reversed  $N$-tuple  $\overline{\mathbf{v}}  :=
(v_N, \dots,  v_1)$ and the opposite  $N$-tuple $-\mathbf{v}
:=  (-v_1, \dots,  -v_N)$,  and  let $\mathcal{T}  \subseteq
\mathcal{V}^{\underline{N}}$  denote the  set  of tuples  in
which  each pair  of centrally  symmetric vectors  appear in
symmetric positions within the tuple, that is
\begin{align*}
  \mathcal{T}      :=       \left\{      \mathbf{v}      \in
  \mathcal{V}^{\underline{N}}  \;  \Big|  \;  -\mathbf{v}  =
  \overline{\mathbf{v}} \right\}.
\end{align*}
Moreover,  for any  $v \in  \mathcal{V}$ let  $\mathcal{T}_v
\subseteq  \mathcal{V}^{\underline{N}}$  denote the  set  of
tuples with fixed point $v_N = v$, that is
\begin{align*}
  \mathcal{T}_v      :=      \left\{     \mathbf{v}      \in
  \mathcal{V}^{\underline{N}} \; \Big| \; v_N = v \right\}.
\end{align*}
The following result holds.

\begin{lmm}[Symmetries]
  For any given  set $\mathcal{V}$ of $N$  Pauli vectors, if
  $\mathcal{V}$ is centrally symmetric or vertex transitive,
  there exists a tuple $\mathbf{v}$ attaining the maximum in
  Eq.~\eqref{eq:optimization}  such   that  $\mathbf{v}  \in
  \mathcal{T}$    or    $\mathbf{v}   \in    \mathcal{T}_v$,
  respectively.   Moreover,  if $\mathcal{V}$  is  centrally
  symmetric  and vertex  transitive,  there  exists a  tuple
  $\mathbf{v} \in \mathcal{T}  \cap \mathcal{T}_v$ attaining
  the maximum in Eq.~\eqref{eq:optimization}.
\end{lmm}

\begin{proof}
  Let  us   consider  the  centrally  symmetric   case.   By
  specializing Lemma  4 of Ref.~\cite{DBK22} to  the case of
  qubit ensemble with  uniform probability distribution, one
  has
  \begin{align*}
    \max_{\mathbf{v} \in \mathcal{V}^{\underline{N}}} \left|
    \left\langle  \mathbf{v}   \right\rangle  \right|_2^2  =
    \max_{\mathbf{v}  \in  \mathcal{T}} \left|  \left\langle
    \mathbf{v} \right\rangle \right|_2^2,
  \end{align*}
  that  proves the  statement.

  Let us consider the  vertex transitive case. By definition
  of  vertex transitivity,  one has  that the  range of  the
  squared norm $| \cdot |_2^2$  is unchanged if the range is
  restricted    from     $\mathcal{V}^{\underline{N}}$    to
  $\mathcal{T}_v$, that is
  \begin{align*}
    \left|   \langle   \mathcal{V}^{\underline{N}}   \rangle
    \right|_2^2  =  \left|   \langle  \mathcal{T}_v  \rangle
    \right|_2^2,
  \end{align*}
  that proves the statement.

  Let  us  consider  the   centrally  symmetric  and  vertex
  transitive  case. Due  to central  symmetry, there  exists
  tuple    $\mathbf{v}$    attaining    the    maximum    in
  Eq.~\eqref{eq:optimization}  such   that  $\mathbf{v}  \in
  \mathcal{T}$.   Due to  vertex  transitivity, every  tuple
  $\mathbf{v}$  is  unitarily  equivalent   to  a  tuple  in
  $\mathcal{T}_v$. Hence the statement is proved.
\end{proof}

The set $\mathcal{T}$ can be  generated as follows.  Let set
$\mathcal{V}' \subseteq  \mathcal{V}$ be  any subset  of the
set $\mathcal{V}$ of states, containing one element for each
pair of centrally  symmetric elements, that is,  for each $v
\in \mathcal{V}$  either $v$  or $-v$ is  in $\mathcal{V}'$,
but not both.  Let also set $\overline{\mathcal{V}}'$ be the
complement     of     set     $\mathcal{V}$,     that     is
$\overline{\mathcal{V}'}     :=    \mathcal{V}     \setminus
\mathcal{V}'$ or  equivalently $\overline{\mathcal{V}'}  = -
\mathcal{V}'$. First,  we show  that sets  $\mathcal{T}$ and
$\{ -1,  +1 \}^{N/2}  \times \mathcal{V}'^{\underline{N/2}}$
are in one-to-one correspondence. For any tuple $\mathbf{v}$
in  $\mathcal{T}$, one  has that  tuples $\boldsymbol{\tau}$
and $\mathbf{v}'$ given by
\begin{align*}
  \tau_i := \begin{cases}
    +1 & \textrm{ if } v_i \in \mathcal{V}',\\
    -1 & \textrm{ otherwise},
  \end{cases}
\end{align*}
and
\begin{align*}
  v_i := \begin{cases}
    v_i & \textrm{ if } v_i \in \mathcal{V}',\\
    -v_i & \textrm{ otherwise},
  \end{cases}
\end{align*}
for any  $i$ in  $\{ 1, \dots,  N/2 \}$, are  in $\{  -1, +1
\}^{N/2}$        and       $\mathcal{V}'^{\underline{N/2}}$,
respectively. Vice-versa,  for any  tuple $\boldsymbol{\tau}
\in  \{ -1,  +1  \}^{N/2}$ and  any  tuple $\mathbf{v}'  \in
\mathcal{V}'^{\underline{N/2}}$ one has that tuple
\begin{align}
  \label{eq:construction}
  \mathbf{v} :=  \boldsymbol{\tau} \odot  \mathbf{v}' \oplus
  \overline{ -\boldsymbol{\tau} \odot \mathbf{v}'}
\end{align}
is in  $\mathcal{T}$, where $\odot$ and  $\oplus$ denote the
Hadamard (that is, element-wise)  product and the direct sum
(that is, concatenation) of tuples, respectively.  Moreover,
due  to  Eq.~\eqref{eq:construction}   any  element  of  set
$\mathcal{T}$  can  be   generated  from  the  corresponding
element     of     set      $\{-1,     +1\}^{N/2}     \times
\mathcal{V}'^{\underline{N/2}}$   in  constant   time.   The
complexities~\cite{Sed77,Knu11}  of generating  the elements
of     set     $\{-1,      1\}^{N/2}$     and     of     set
$\mathcal{V}'^{\underline{N/2}}$   are    $O(2^{N/2})$   and
$O((N/2)!)$,   respectively.   Moreover,   through  a   Gray
code~\cite{Knu11} it is possible to iteratively generate all
tuples $\boldsymbol{\tau}$ in $\{  -1, 1 \}^{N/2}$ such that
each  new   tuple  $\boldsymbol{\tau}'$  differs   from  the
previous tuple  $\boldsymbol{\tau}$ by  a single  sign flip.
If  the  sign  flip  occurs in  the  $i$-th  position,  from
Eq.~\eqref{eq:avg}  it  immediately  follows  that  function
$\langle \mathbf{v} \rangle$ is updated as follows
\begin{align}
  \label{eq:flip}
  \left\langle  \mathbf{v}'  \right\rangle_N =  \left\langle
  \mathbf{v} \right\rangle_N - 4 \left( 2  i - N + 1 \right)
  v_i,
\end{align}
that represents a constant-time  update.  Since for even $N$
one has  $N!!  =  2^{N/2} (  N/2 )!$ (with  $N!! :=  N (N-2)
(N-4)  \dots$  we  denote the  double  factorial  function),
central symmetry  can be exploited to  reduce the complexity
of Algorithm~\ref{algo:gw} by a factor of $(N - 1)!!$.

The  set  $\mathcal{T}_{v}$  can be  generated  as  follows.
First,  notice that  sets $\mathcal{T}_v$  and $(\mathcal{V}
\setminus   v)^{\underline{N  -   1}}$  are   in  one-to-one
correspondence.     For    any   tuple    $\mathbf{v}$    in
$\mathcal{T}_v$, one  has that its  restriction $\mathbf{v}'
:= \mathbf{v} |_{\{1,  \dots, N - 1\}}$  is in $(\mathcal{V}
\setminus  v)^{\underline{N  -  1}}$.  Vice-versa,  for  any
tuple     $\mathbf{v}'$    in     $(\mathcal{V}    \setminus
v)^{\underline{N - 1}}$, one has that its extension
\begin{align}
  \label{eq:extension}
  v_i  := \begin{cases}  v'_i &  \textrm{  if }  i \in  \{1,
    \dots, N - 1\},\\ v & \textrm{ otherwise},
  \end{cases}
\end{align}
is     in     $\mathcal{T}_v$.      Moreover,     due     to
Eq.~\eqref{eq:extension}, any element of set $\mathcal{T}_v$
can  be   generated  from   the  corresponding   element  of
$(\mathcal{V} \setminus v)^{\underline{N  - 1}}$ in constant
time.   Since all  the elements  of $(\mathcal{V}  \setminus
v)^{\underline{N  -  1}}$  can   be  generated  in  $(N-1)!$
steps~\cite{Sed77,Knu11},   vertex   transitivity   can   be
exploited      to     reduce      the     complexity      of
Algorithm~\ref{algo:gw} by a factor of $N$.

The set $\mathcal{T} \cap \mathcal{T}_v$ can be generated by
concatenating   the   two    previous   methods.    Assuming
$\mathcal{V}$ is centrally  symmetric and vertex transitive,
for  any $v  \in \mathcal{V}$,  let $\mathcal{V}'  \subseteq
\mathcal{V} \setminus \{  v, -v\}$ be such that  for any $v'
\in \mathcal{V}$  one has $-v' \not\in  \mathcal{V}'$. Then,
sets $\mathcal{T} \cap \mathcal{T}_v$  and $\{-1, 1\}^{N/2 -
  1}  \times  \mathcal{V}'^{\underline{N/2  -  1}}$  are  in
one-to-one  correspondence.    The  correspondence   can  be
explicitly  built  as  before.   Moreover,  any  element  of
$\mathcal{T} \cap  \mathcal{T}_v$ can be generated  from the
corresponding  element  of  $\{-1,   1\}^{N/2  -  1}  \times
\mathcal{V}'^{\underline{N/2 - 1}}$  in constant time. Since
the  complexities~\cite{Sed77,  Knu11}   of  generating  the
elements    of   sets    $\{-1,   1\}^{N/2    -   1}$    and
$\mathcal{V}'^{\underline{N/2 - 1}}$  are $O(2^{N/2-1})$ and
$O((N/2 -  1)!)$, respectively, central symmetry  and vertex
transitivity can  be exploited  to reduce the  complexity of
Algorithm~\ref{algo:gw}  by a  factor of  $N (N-1)!!$,  that
represents a  more-than-quadratic speedup. By  denoting with
$\operatorname{nextGray}$  the  function  that  returns  the
index  $k$ that  needs to  be flipped  to generate  the next
tuple   according   to   the  Gray   code~\cite{Knu11}   and
$\operatorname{false}$  if one  cycle of  the code  has been
completed  and the  algorithm  is back  to  the first  tuple
produced, we arrive at Algorithm~\ref{algo:gwsym}.

\begin{algorithm}
\caption{Guesswork (more-than-quadratic speedup)}
\label{algo:gwsym}
\begin{algorithmic}
  \REQUIRE  finite centrally symmetric and vertex invariant set  $\mathcal{V}$  of $N$ Pauli  vectors
  \ENSURE $G_{\min} ( \mathcal{V} ) = (N + 1 - \sqrt{g}/N)/2$
  \STATE $\mathbf{v} \in \mathcal{V}^{\underline{N/2-1}}$ s.t. $\forall v \in \mathbf{v}$ one has $-v \not\in \mathbf{v}$
  \STATE $\overline{v} \gets \langle \mathbf{v} \rangle$
  \STATE  $g \gets 0$
  \WHILE {$(i, j) \gets \operatorname{nextJT} ( \mathbf{v} )$}
  \STATE $v_i \leftrightarrow v_j$
  \STATE $\overline{v} \gets \overline{v}  + 2 (j-i) (v_i - v_j)$
  \STATE $\boldsymbol{\tau} \in \{ -1, 1 \}^{N/2-1}$
  \WHILE {$k \gets \operatorname{nextGray} ( \boldsymbol{\tau} )$}
  \STATE $v_k \gets - v_k$
  \STATE $\overline{v} \gets \overline{v} - 4 ( 2 k - N + 1) v_k$  
  \STATE $g \gets \max ( g, | \overline{v} |_2^2 )$
  \ENDWHILE
  \ENDWHILE
  \RETURN $g$
\end{algorithmic}
\end{algorithm}

Table~\ref{tab:complexity}  summarizes the  results of  this
section.

\begin{table}[h!]
    \begin{center}
    \caption{Complexity      of     Algorithms~\ref{algo:gw}
      and~\ref{algo:gwsym} for the  exact computation of the
      minimum   guesswork  of   any  given   qubit  ensemble
      $\mathcal{V}$  with uniform  probability distribution,
      as a function of the symmetries of $\mathcal{V}$.}
    \label{tab:complexity}
    \begin{tabular}{|l|c|c|}
      \hline
      Symmetries & Complexity & Speedup (with respect to no symmetries)\\
      \hline
      No symmetries & $O(N!)$ & $1$\\
      Central symmetry & $O(N!!)$ & $(N - 1)!!$\\
      Vertex transitivity & $O((N-1)!)$ & $N$\\
      Central symm. \& vertex trans. & $O((N-2)!!)$ & $N (N - 1)!!$\\
      \hline
    \end{tabular}
  \end{center}
\end{table}

\section{Explicit examples}
\label{sect:examples}

In   this  section   we  apply   Algorithms~\ref{algo:gwsym}
and~\ref{algo:sym}   to   compute  the   exact   closed-form
expression  for   the  minimum  guesswork  of   regular  and
quasi-regular  qubit ensembles  up to  twenty-four vertices.
We also provide~\cite{Dal21} a  C language implementation of
such algorithms.

First, let us discuss ensembles of Pauli vectors whose
coordinates can be represented (up to a scaling) by the ring
of integers.  These are the tetrahedron, octahedron, cube,
truncated tetrahedron, cuboctahedron, and truncated
octahedron. The values of the guesswork for tuples of Pauli
vectors on the ring of integers are reported in
Table~\ref{tab:integer}.

\begin{table}[h!]
    \begin{center}
    \caption{Exact  closed-form  expression and  approximate
      numerical value  of the  minimum guesswork  of regular
      and quasi-regular  tuples of qubit states  on the ring
      of integers, as given by Algorithm~\ref{algo:gwsym}.}
    \label{tab:integer}
    \begin{tabular}{|l|c|c|c|}
      \hline
      $\mathcal{V}$ & $N$ & $g$ & $G_{\min}$\\
      \hline
      Tetrahedron & $4$ & $\frac{80}{3}$ & $\sim 1.8545$\\
      Octahedron & $6$ & $140$ & $\sim 2.5140$\\
      Cube & $8$ & $\frac{1344}{3}$ & $\sim 3.1771$\\
      Truncated tetrahedron & $12$ & $\frac{25168}{11}$ & $\sim 4.5070$\\
      Cuboctahedron & $12$ & $\frac{4560}2$ & $\sim 4.5104$\\
      Truncated octahedron & $24$ & $\frac{183440}5$ & $\sim 8.5096$\\
      \hline
    \end{tabular}
  \end{center}
\end{table}

Second,  let us  discuss  ensembles of  Pauli vectors  whose
coordinates can be represented (up to a scaling) by the ring
\begin{align*}
  p_k \left( \mathbf{z} \right) := z_0 + \sqrt{k} z_1,
\end{align*}
where $\mathbf{z} = (z_0 , z_1) \in \mathbb{Z}^2$ and $k$ is
a positive non square integer  constant that only depends on
the  ensemble.   These  are the  icosahedron,  dodecahedron,
truncated   cube,   and    rhombicuboctaedron.    To   apply
Algorithm~\ref{algo:sym}  on   a  machine   that  implements
integer arithmetic, we need  to derive integer formulae for
the arithmetic operations of sum, difference, multiplication
by an integer, and square.  They are clearly given by
\begin{align*}
  p_k \left(  \mathbf{z} \right) \pm p_k  \left( \mathbf{z}'
  \right) = p_k \left( z_0 \pm z'_0, z_1 \pm z'_1 \right),
\end{align*}
\begin{align*}
  z  \;  p_k  \left(  \mathbf{z}  \right)  =  p_k  \left(  z
  \mathbf{z} \right),
\end{align*}
and
\begin{align*}
  p_k \left(  \mathbf{z} \right)^2  = p_k  \left( z_0^2  + k
  z_1^2 , 2 z_0 z_1 \right).
\end{align*}
We  also need  an integer  formula to  compare numbers.   By
direct inspection one has $p_k (  \mathbf{z} ) \ge 0$ if and
only if
\begin{align*}
  \left( z_0 \ge 0 \textrm{ and  } z_0^2 \ge k z_1^2 \right)
  \textrm{ or } \left( z_1 \ge  0 \textrm{ and } z_0^2 \le k
  z_1^2 \right).
\end{align*}
The values of  the guesswork for tuples of  Pauli vectors on
this ring are reported in Table~\ref{tab:golden}.

\begin{table}[h!]
    \begin{center}
    \caption{Exact  closed-form  expression and  approximate
      numerical value  of the  minimum guesswork  of regular
      and quasi-regular  tuples of qubit states  on the ring
      $z_0     +    \sqrt{k}     z_1$,    as     given    by
      Algorithm~\ref{algo:gwsym}.}
    \label{tab:golden}
    \begin{tabular}{|l|c|c|c|}
      \hline
      $\mathcal{V}$ & $N$ & $g$ & $G_{\min}$\\
      \hline
      Icosahedron & $12$ & $\frac{16544 + 7392 \sqrt{5}}{10 + 2 \sqrt{5}}$ & $\sim 4.5081$\\
      Dodecahedron & $20$ & $\frac{106272 + 47456 \sqrt{5}}{12}$ & $\sim 7.1741$\\
      Truncated cube & $24$ & $\frac{47040 + 23168 \sqrt{2}}{5 - 2 \sqrt{2}}$ & $\sim 8.5062$\\
      Rhombicuboctahedron & $24$ & $\frac{146128 + 100128 \sqrt{2}}{5 + 2 \sqrt{2}}$ & $\sim 8.5059$\\
      \hline
    \end{tabular}
  \end{center}
\end{table}

\section{An exact symmetries finding algorithm}
\label{sect:symmetries}

In this section we present an algorithm that, upon the input
of  any  arbitrary-dimensional   complex  point  set,  after
finitely  many-steps  outputs  its  exact  symmetries.   The
complexity of our  algorithm is polynomial in  the number of
points.

Previous  works~\cite{Hig85, WWV85,  BK02, KR16}  approached
the problem of finding the symmetries of any given point set
(and the  related problem of  testing the congruence  of two
sets) from the geometric viewpoint,  that is, by looking for
unitary transformations that act as permutations of the set.
As  a  consequence, previous  symmetries-finding  algorithms
depend  on the  full  field structure  (in particular,  they
depend on  the arithmetic  operation of division).   To this
aim, they assume the \textit{real computational model}, that
is, an  unphysical machine that  can exactly store  any real
number  and can  exactly perform  arithmetic, trigonometric,
and other functions over reals in finite time.

We instead approach the  symmetries-finding problem from the
viewpoint  of   combinatorics,  that  is,  by   looking  for
permutations of the set that act as unitary transformations.
In fact,  by using well-known  results on Gram  matrices, we
avoid  explicitly   dealing  with   unitary  transformations
altogether.   This  way,  we  present  a  symmetries-finding
algorithm   (that  can   also   be   trivially  adapted   to
congruence-testing)  that only  depends on  the weaker  ring
structure  (that  is, division  is  not  assumed).  Ours  is
therefore  an  \textit{integer   computational  model}  that
solely assumes the  ability to store integer  numbers and to
perform additions  and multiplication  in finite  time, thus
allowing  us to  achieve closed-form  analytical results  on
physical machines.

The factorial  growth of the  number of permutations  in the
cardinality $N$ of the set dooms to factorial complexity any
algorithm based  on a naive exhaustive  search.  However, by
exploiting a  well-known rigidity property of  simplices, we
show that without  loss of generality it  suffices to search
over  a   polynomial-sized  subset  of   permutations.   The
complexity of our  symmetries-finding algorithm is therefore
$O(N^{d+2})$, where $d$ denotes the dimension of the complex
space.

For   any    given   arbitrary-dimensional    spanning   set
$\mathcal{V}$  of complex  vectors, we  denote with  $\Sym (
\mathcal{V} )$  the group of permutations  of $\mathcal{V}$.
A permutation  $\sigma$ in  $\Sym(\mathcal{V})$ is  called a
\textit{geometric symmetry} (in  the following, symmetry for
short)  of $\mathcal{V}$  if and  only there  exists unitary
transformation $\mathcal{U}$ such that
\begin{align}
  \label{eq:symmetry}
  \sigma  \left(  \mathbf{v}  \right) =  \mathcal{U}  \left(
  \mathbf{v} \right),
\end{align}
where  $\mathbf{v} \in  \mathcal{V}^{\underline{N}}$ denotes
an  $N$-tuple  on  $\mathcal{V}$  without  repetitions.   We
denote  the group  of all  symmetries of  $\mathcal{V}$ with
$\Geom  (   \mathcal{V}  )$.  Notice  that   the  fact  that
$\mathcal{V}$ is a spanning  set guarantees that the mapping
between     $\sigma$     and    $\mathcal{U}$     satisfying
Eq.~\eqref{eq:symmetry} is bijective.

Since   the  computation   of  the   unitary  transformation
$\mathcal{U}$ in  Eq.~\eqref{eq:symmetry} requires divisions
in general (take for example $\mathbf{v}$ to be the vertices
of a square  and $\mathcal{U}$ to correspond  with a $\pi/2$
rotation), we  need an approach where  $\mathcal{U}$ remains
implicit. It is a well-known fact that two tuples of vectors
are unitarily  related if  and only  if their  Gram matrices
coincide, where  the Gram  matrix $G (  \mathbf{v} )$  of an
$N$-tuple $\mathbf{v}  := ( \mathbf{v}_i )_i$  of vectors is
the $N \times N$ matrix whose  $(i, j)$-th entry is given by
the inner product $\mathbf{v}_i \cdot \mathbf{v}_j$, that is
\begin{align}
  \label{eq:gram}
  \left[  G  \left(  \mathbf{v}  \right)  \right]_{i,  j}  =
  \mathbf{v}_i \cdot \mathbf{v}_j.
\end{align}
Hence,  tuples  $\sigma  ( \mathbf{v}  )$  and  $\mathbf{v}$
satisfy Eq.~\eqref{eq:symmetry} if and only if
\begin{align}
  \label{eq:symmetry_gram}
  G  \left( \sigma  \left(  \mathbf{v} \right)  \right) =  G
  \left( \mathbf{v} \right).
\end{align}

Since  Eq.~\eqref{eq:gram} can  be clearly  computed without
divisions,   this  observation   immediately   leads  to   a
division-free exact symmetries-finding  algorithm through an
exhaustive search  over the  set $\Sym  ( \mathcal{V}  )$ of
permutations.  The  complexity of  the exhaustive  search is
given by  the product  of the  complexity of  generating all
permutations  $\sigma  \in  \Sym(  \mathcal{V}  )$  and  the
complexity of computing and comparing the corresponding Gram
matrices.   According  to  Eq.~\eqref{eq:gram},  the  latter
complexity  is  $O(N^2)$.   Hence,  the  complexity  of  the
algorithm    is   $O(N!     N^2)$.     By   denoting    with
$\operatorname{nextJT}$ the  function that returns  the pair
of indexes $(i, j)$ that need  to be swapped to generate the
next   tuple  according   to  Heap's   or  Johnson-Trotter's
algorithm~\cite{Knu11} and $\operatorname{false}$ if the end
of   the  algorithm   has   been  reached,   we  arrive   at
Algorithm~\ref{algo:symfact}.

\begin{algorithm}[h!]
  \caption{Symmetries finding (exhaustive search)}
  \label{algo:symfact}
  \begin{algorithmic}
    \REQUIRE $d$-dimensional spanning set $\mathcal{V}$ of $N$ complex vectors
    \ENSURE $\mathcal{S} = \Geom ( \mathcal{V} ) ( \mathbf{v} )$, for some $\mathbf{v} \in \mathcal{V}^{\underline{N}}$
    \STATE $\mathbf{v} \in \mathcal{V}^{\underline{N}}$
    \STATE $\mathbf{v}' \gets \mathbf{v}$
    \STATE $\mathcal{S} \gets \emptyset$
    \WHILE {$(i, j) \gets \operatorname{nextJT} ( \mathbf{v} )$}
    \STATE $v_i' \leftrightarrow v_j'$
    \IF {$G(\mathbf{v}') = G(\mathbf{v})$}
    \STATE $\mathcal{S} \gets \mathcal{S} \cup \mathbf{v}'$ 
    \ENDIF
    \ENDWHILE
    \RETURN $\mathcal{S}$
  \end{algorithmic}
\end{algorithm}

We   proceed    now   to    improve   the    complexity   of
Algorithm~\ref{algo:symfact}      from     factorial      to
polynomial. Let $d$-tuple  $\mathbf{e} := ( e_i  )_{i = 1}^d
\in  \mathcal{V}^{\underline{d}}$  on $\mathcal{V}$  without
repetitions be  a basis,  that is,  the determinant  $\det G
(\mathbf{e})$ of its Gram matrix is non null.  Division-free
algorithms for the computation of the determinant are known;
for a  particularly simple one, see  Ref.~\cite{Bir11}.  Due
to Eq.~\eqref{eq:symmetry_gram},  a necessary  condition for
any permutation  $\sigma$ of $\mathcal{V}$ to  be a symmetry
is that
\begin{align}
  \label{eq:basis}
  G  \left( \sigma  \left(  \mathbf{e} \right)  \right) =  G
  \left( \mathbf{e} \right).
\end{align}
For any $d$-tuple  $\mathbf{e}'$ satisfying $G( \mathbf{e}')
=  G(  \mathbf{e} )$,  the  permutation  $\sigma$ such  that
Eq.~\eqref{eq:symmetry_gram} holds, if  it exists, is unique
and can be explicitly derived as follows.

For some  basis $\mathbf{e}$ and  any vectors $v_0,  v_1 \in
\mathbb{C}^d$  we say  $v_0 \prec_{\mathbf{e}}  v_1$ if  and
only if
\begin{align*}
  e_k \cdot \left( v_0 - v_1 \right) \le 0,
\end{align*}
where $k$  is the minimum  over $\{1, \dots, d\}$  such that
$e_k  \cdot  (  v_0 -  v_1  )$  is  not  null, and  we  call
$\mathbf{e}$-order      the       order      induced      by
$\prec_{\mathbf{e}}$.   Since   $\mathbf{e}$  is   a  basis,
$\mathbf{e}$-order   is   total.     Let   $\mathbf{v}   \in
\mathcal{V}^{\underline{N}}$  be   the  $\mathbf{e}$-ordered
$N$-tuple  on  $\mathcal{V}$  without repetitions.   Due  to
Eq.~\eqref{eq:symmetry}, for  any $k$  in $\{1,  \dots, d\}$
and any  $i$ in  $\{1, \dots, N\}$,  the inner  product $e_k
\cdot v_i$  equals the inner  product $\sigma ( e_k  ) \cdot
\sigma( v_i )$. Hence, tuple  $\sigma ( \mathbf{v} )$ is the
$\sigma( \mathbf{e} )$-ordered $N$-tuple of all the elements
of $\mathcal{V}$  without repetitions.  This  explicitly and
uniquely identifies permutation $\sigma$.

The complexity of  the algorithm is given by  the product of
the   complexity    of   generating   all    $d$-tuples   in
$\mathcal{V}^{\underline{d}}$   and    the   complexity   of
processing  each  tuple.   Since  the  combinations  of  $d$
elements  out  of  $N$  are  $N \choose  d$,  and  for  each
combination there  are $d!$ differently ordered  tuples, the
complexity~\cite{Sed77, Knu11} of  generating the $d$-tuples
is $O( {N \choose d})$, that  is, a polynomial of degree $d$
in $N$.  The complexity of  processing each tuple is the sum
of  the   complexity  of   generating  $\mathbf{r}'$-ordered
$N$-tuple  $\mathbf{v}'$, computing  the  Gram  matrix $G  (
\mathbf{v}' )$,  and comparing it  with $G (  \mathbf{v} )$,
hence  it  is  $O(N^2)$.    Hence,  the  complexity  of  the
algorithm   is   $O(N^{d   +   2})$.    By   denoting   with
$\operatorname{nextChase}$  the  function that  returns  the
pair of indexes $(i, j)$ that need to be swapped to generate
the    next   combination    according   to    the   Chase's
sequence~\cite{Knu11} and $\operatorname{false}$  if the end
of   the  algorithm   has   been  reached,   we  arrive   at
Algorithm~\ref{algo:sym}.

\begin{algorithm}[h!]
  \caption{Symmetries finding (polynomial time)}
  \label{algo:sym}
  \begin{algorithmic}
    \REQUIRE $d$-dimensional spanning set $\mathcal{V}$ of $N$ complex vectors
    \ENSURE $\mathcal{S} = \Geom ( \mathcal{V} ) ( \mathbf{v}' )$, for some $\mathbf{v}' \in \mathcal{V}^{\underline{N}}$
    \STATE $\mathbf{v} \gets \textrm{$\mathbf{e}$-ordered $N$-tuple in } \mathcal{V}^{\underline{N}}$
    \STATE $\mathcal{S} \gets \emptyset$

    \WHILE {$(i, j) \gets \operatorname{nextChase}(\mathbf{v}, \mathbf{e}')$}
    \STATE $e'_i \gets v_j$
    \WHILE {$(k, l) \gets \operatorname{nextJT}( \mathbf{e}' )$}
    \STATE $e_k' \leftrightarrow e_l'$
    \IF {$G (\mathbf{e}') =  G (\mathbf{e})$}
    \STATE $\mathbf{v}' \gets \textrm{$\mathbf{e}'$-ordered $N$-tuple in } \mathcal{V}^{\underline{N}}$
    \IF {$G( \mathbf{v}') = G ( \mathbf{v})$}
    \STATE $\mathcal{S} \gets \mathcal{S} \cup \mathbf{v}'$ 
    \ENDIF
    \ENDIF
    \ENDWHILE
    \ENDWHILE
    \RETURN $\mathcal{S}$
  \end{algorithmic}
\end{algorithm}

\section{Conclusion}
\label{sect:conclusion}

We showed  that the  computation of  the guesswork  of qubit
ensembles with uniform  probability distribution corresponds
to   a  quadratic   assignment  problem.   We  presented   a
division-free algorithm for the exact analytical computation
of the  guesswork with a more-than-quadratic  speedup in the
presence of  symmetries, that is, for  the three-dimensional
analog of the maximization  version of the turbine-balancing
problem.   As examples,  we computed  the exact  closed-form
expression for  the guesswork  of regular  and quasi-regular
ensembles of qubit states.

\section{Acknowledgments}

This   work  is   dedicated   to  the   memory  of   Takeshi
Koshiba. M.~D.  acknowledges support  from the Department of
Computer  Science and  Engineering, Toyohashi  University of
Technology,  the MEXT  Quantum Leap  Flagship Program  (MEXT
Q-LEAP)  Grant  No.   JPMXS0118067285,  JSPS  KAKENHI  Grant
Number JP20K03774,  and the  International Research  Unit of
Quantum Information, Kyoto  University.  F.  B. acknowledges
support from MEXT Quantum Leap Flagship Program (MEXT QLEAP)
Grant No.  JPMXS0120319794;  from MEXT-JSPS Grant-in-Aid for
Transformative  Research Areas  (A) “Extreme  Universe”, No.
21H05183;  from JSPS  KAKENHI  Grants No.  20K03746 and  No.
23K03230.


\begin{thebibliography}{}
\bibitem{Sed77} R. Sedgewick, \textit{Permutation generation
  methods}, Computing Surveys \textbf{9}, 137 (1977). Notice
  that in Algorithm 3a therein, the third line from the last
  should read ``c[i] := 1; d[i] := not d[i]; i := i + 1;''.
\end{thebibliography}
\end{document}